\documentclass[10pt,doublesided,doublecolumn]{IEEEtran}
\usepackage{graphicx,epsfig,amsmath,amssymb,bm}
\usepackage{amssymb,balance}
\usepackage{amsmath}
\usepackage{amsthm} 
\usepackage{tikz}
\usepackage{amssymb}
\usepackage{algorithmicx}
\usepackage{algpseudocode}
\newtheorem{lemma}{Lemma}

\newtheorem{theorem}{Theorem}

\def\LSB{\left[}        
\def\RSB{\right]}       
\def\LB{\left(}         
\def\RB{\right)}        

\newfont{\bbb}{msbm10 scaled 500}

\newfont{\bb}{msbm10 scaled 1100}

\newcommand{\Em}{{\bf E}}

\newcommand{\Pm}{{\bf P}}

\newcommand{\piv}{\hbox{\boldmath$\pi$}}

\newcommand{\defines}{{\,\,\stackrel{\scriptscriptstyle \bigtriangleup}{=}\,\,}}

\def\argmin{\operatornamewithlimits{arg\,min}}

\newtheorem{example}{Example}

\newcommand{\beqa}{\begin{eqnarray}}
\newcommand{\eeqa}{\end{eqnarray}}
\newcommand{\dsp}{\displaystyle}

\begin{document}

\title{Cooperation and Storage Tradeoffs in Power-Grids with Renewable Energy Resources}
\author{Subhash~Lakshminarayana,~\IEEEmembership{Member, IEEE}, Tony Q. S. Quek,~\IEEEmembership{Senior Member, IEEE} and H. Vincent Poor,~\IEEEmembership{Fellow, IEEE}
\thanks{
Manuscript received 22 Oct, 2013; accepted 28 Mar, 2014. 
This paper was presented in part at the IEEE INFOCOM Workshop on Communications and Control for Smart Energy Systems, Toronto, Canada, April, 2014.
This research was supported, in
part, by the SUTD-MIT International Design Centre under Grant
IDSF1200106OH.
\newline S. Lakshminarayana is with the Singapore University of Technology and Design, Singapore (email: subhash@sutd.edu.sg). 
\newline T.Q.S. Quek is with the Singapore University of Technology and Design, Singapore and the Institute for Infocomm Research, Singapore (email: tonyquek@sutd.edu.sg). 
\newline H. V. Poor is
with the Department of Electrical Engineering, Princeton University, Princeton, NJ, USA (email: poor@princeton.edu). 
}} 

\maketitle

\begin{abstract}
One of the most important challenges 
in smart grid systems 
is the integration of renewable energy 
resources into its design. 
In this work, two different 
techniques to mitigate the time varying 
and intermittent nature of renewable energy
generation are considered. The first one is the use of storage,
which smooths out the fluctuations in the renewable energy generation across time.
The second technique is the concept of
distributed generation combined with cooperation by exchanging
energy among the distributed sources.
This technique averages out the variation in energy production across space. This paper analyzes the trade-off between these two techniques.
The problem is formulated as a stochastic optimization
problem with the objective of minimizing 
the time average cost of energy exchange within the grid.
First, an analytical model
of the optimal cost is provided by investigating the steady state 
of the system for some specific scenarios.
Then, an algorithm to solve the cost minimization problem using the technique of Lyapunov optimization is developed
and results for the performance 
of the algorithm are provided.
These results show that in the presence of limited
storage devices, the grid can benefit greatly from cooperation,
whereas in the presence of large storage capacity, cooperation
does not yield much benefit. Further, it is observed that most
of the gains from cooperation can be obtained by 
exchanging energy only among a few energy harvesting sources.
\end{abstract}


\section{Introduction}
Renewable energy provides a greener alternative
to traditional fossil fuel based electric
power generation.
Thus, there has been significant emphasis on integration 
of renewable energy into smart grid design \cite{VariableGen2009}.
However, a significant challenge lies in 
the inherently stochastic 
and intermittent nature of renewable energy production.
A popular technique to compensate 
for this is the use of expensive fast-ramping fuel-based generators as a back-up. However, with greater penetration
of renewable energy this technique is no longer cost effective \cite{HartCarbonEm2011}.
As a result, there is a need for more cost effective solutions such as the use of energy storage \cite{RobertsGridPower2009},
and load scheduling by demand-response \cite{FERC2006}.

Prior work on design and analysis of 
renewable energy with storage includes
\cite{ChandyLow2010},
which formulates the finite horizon optimal 
power flow problem with storage as a
convex optimization problem.
This work shows that for the special case of a single generator and single load, the optimal policy is
to charge the battery at the beginning 
and discharge towards the end of the 
time horizon.
\cite{SuGamalStorage2011} and \cite{KostTass2011} formulate
the problem as a dynamic programming problem and
derive threshold based control policies for battery
charging and discharging decisions.
In \cite{UrgaonkarSigmetric2011}, storage is used as a 
means to reduce the time average electricity bills 
in data center applications. Using a Lyapunov optimization
based approach, this work shows that increasing the storage capacity
results in a greater reduction in the electricity bills.
Other relevant works include \cite{PaateroStorage2005} and \cite{BitarRajgopal2011} (and references within).
On the other hand, prior work on demand
response includes \cite{KimPoor2011},
which formulates the problem of scheduling the power
consumption as 
a Markov decision problem, where the scheduler
has access to the past and the current prices, but only statistical
knowledge about future prices.
It is shown that incorporating the statistical knowledge into the scheduling policies can result in significant savings.
\cite{LiChenLowPEGM2011}
considers the problem of demand response in a finite 
time domain and solves the problem using
convex optimization based techniques.
The multi-period power procurement 
is handled in \cite{JiangLowCDC2011}, 
where it is solved using stochastic gradient based techniques.
The combination of storage and demand-response 
has been examined in \cite{HuangWalrandRamSGComm2012},
and is solved using Lyapunov optimization based approach
with the conclusion that storage combined with demand response
can give greater cost savings.

One of the other techniques to combat the intermittent nature of renewable energy that has been explored 
relatively less, is the use cooperation in distributed power generation units \cite{WesternWind2009}. 
The idea is to  exploit the averaging effect
produced by diversity in renewable energy production across different geographical areas. By enabling cooperation, areas that have excess production can transfer
energy to areas that are deficient. 
For example, studies have been conducted by monitoring the renewable energy production across 5 different states in the United States (Arizona, Colorado, Nevada, New Mexico, and Wyoming).  
These studies have shown that while the variability of the load greatly increases with an increase in penetration of renewable energy in the individual states, aggregating the diverse
renewable resources over these geographic areas
leads to only a slight increase in the load variability \cite{WesternWind2009}.
This also leads to a very substantial reduction in the operating cost
of the grid as well (\$2 billion in this case).
In terms of analytical results, the impact of aggregation of wind power has been considered
in the framework of coalitional game theory in
\cite{BitarCoalition2011} and \cite{SaadHanPoorICC2011}, where it is shown that independent wind producers can 
benefit by aggregating their harvested energy.
Distributed energy production
has also been studied within the framework 
of \emph{micro-grids} (MGs) \cite{Microgrids2007}, with a focus on
distributed storage and decentralized control of 
MG networks \cite{SmitDistributedStorage2013},\cite{MGDecControl2013}.
Energy sharing among MGs has also been studied 
in \cite{TingZhuMGsEnSharing2013} via simulations, and 
shown to reduce energy losses in the network.

While all of the above mentioned works address the issues of storage, demand response and aggregation individually, the combination of the dual averaging effect
produced by storage and cooperation by energy sharing has
not been explored.
The objective of this work is to provide
an analytical framework for studying the trade-off present
between storage and cooperation.
We consider a scenario consisting of 
MGs that are powered by harvesting renewable
energy and are serving their respective loads.
The MGs have finite capacity storage units and can cooperate by transferring energy among themselves. For any excess load, the MGs can borrow energy from the macro-grid. The objective is to minimize
the time average cost of energy exchange within the entire grid.
We first provide an analytical characterization of 
the optimal cost by examining the steady state behavior of
the system for some particular settings.
We then provide an online algorithm to solve
the optimization problem using the technique of
Lyapunov optimization \cite{NeelyBook}.
The control decision to be taken at each time
slot is how to divide the excess renewable energy optimally  between storage and cooperation, and how much
energy is to be borrowed from the macro-grid.
We analyze the optimal cost
as a function of the storage capacity and the number 
of cooperating MGs.
We also investigate the following question: for a given number of cooperating MGs, what is the storage capacity needed
in order to be self sufficient (i.e.
to eliminate the need for energy transfer from the macro-grid).

Our result shows that  
when the storage capacity is low, cooperation
among the MGs yields a significant reduction in
the time average cost of energy exchange.
However, when the MGs have a large
storage capacity, cooperation does 
not yield much benefit. This is because each MGs
can simply store all its excess harvested energy and
use it during the time slots when it is deficient.
Further, most of the gains are obtained by
cooperation among only a few neighboring MGs.

The rest of the paper is organized as 
follows. We present the system model 
and provide the problem formulation of minimizing the time 
average cost of energy transfer across the grid in Section \ref{sec:sysmodel}. We provide the analytical modeling of the optimal
cost for some special cases in Section \ref{sec:analoptcost}.
Then, in Section \ref{sec:LypOpt}, we present an algorithm to solve the time average cost minimization problem using the technique of Lyapunov optimization, and provide
results for the algorithm performance.
Numerical results are presented in Section \ref{sec:numresult}
followed by conclusions in Section \ref{sec:conclusion}.
Finally, the proofs of some results in the 
paper are presented in Appendices A, B and C.
 
\section{System Model}
\label{sec:sysmodel}

\begin{figure*}
\begin{center}
     \begin{scriptsize}
       \def\svgwidth{1.5\columnwidth}
       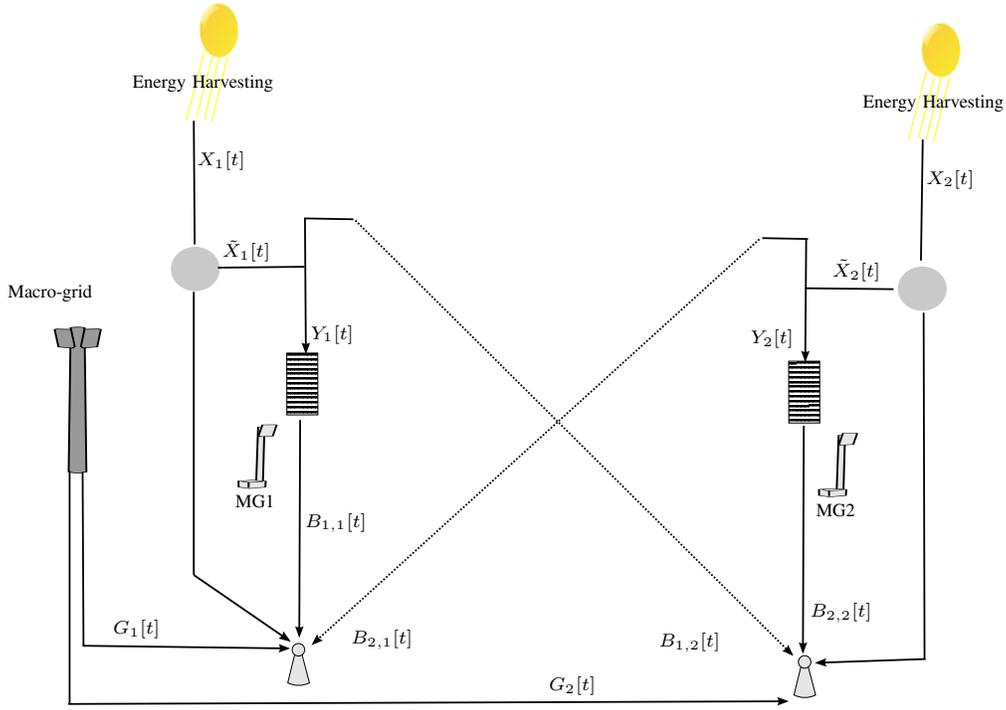
       \caption{Power grid consisting of micro-grids and a macro-grid.}
       \label{fig:MicroGrids}
       \end{scriptsize}
       \end{center}
\end{figure*}
We consider an inter-connected power grid consisting of $N$ MGs and a macro-grid as shown in Figure \ref{fig:MicroGrids}. 
The MGs are capable of harvesting renewable energy (e.g. wind, solar energy etc).
In addition, the MGs are equipped with batteries
in which they can store the harvested energy for future use.

\subsection{Energy Supply, Demand and Distribution Models}
\subsubsection{Energy Generation}
MG$_i$ harvests $X_i[t], \ i = 1,\dots,N$ units of energy during the time slot $t,$ which is the only source of energy generation at the MG. We assume that the harvested energy $X_i[t]$ evolves according to an independent and identically distributed (i.i.d.) random process across time\footnote{The i.i.d. assumption is made for the sake of convenience of illustrations and technical proofs. We note that the algorithm developed in this paper can also be extended to the case where the energy harvesting process in Markovian.}. However, the energy 
harvesting process can be arbitrarily correlated across different MGs. The macro-grid generates energy from convention energy sources.
We assume that the macro-grid has a very large supply of energy
(and do not impose any constraint on its energy generation).

\subsubsection{Load Serving}
MG$_i$ serves a set of users
whose aggregate energy\footnote{With slight abuse of terminology, we use the terms power and energy interchangeably.} demand is $L_i[t]$ units of energy per time slot. The energy demand is bounded as $L_i[t] \leq L_{\max}$, for finite $L_{\max}.$ The energy demand is met in the following manner. 

Firstly, the harvested energy is used to serve the load $L_i[t].$
We consider the two cases as follows: \\
$\bullet$
If $X_i[t] < L_i[t],$ then MG$_i$ uses all the harvested energy to serve its load. The unsatisfied load is denoted by
$\tilde{L}_i[t] = (L_i[t]-X_i[t])^+$ and the MG$_i$ 
does the following to serve the unsatisfied load:
\begin{itemize}
\item[1.] Draw energy stored in its own battery

The MG$_i$ uses $B_{i,i}[t]$ units 
of energy from the energy stored in its own battery to serve the unsatisfied load to its respective users.

\item[2.] Exchange energy among the MGs

In addition, MG$_i$ can borrow 
$B_{j,i}[t]$ units of energy from 
$MG_{j}$ such that $j \neq i$ where $B_{j,i}[t]$ is bounded as$B_{j,i}[t] \leq B^\text{ex}_{\max},$ for some $B^\text{ex}_{\max} < \infty.$ 
Note that $B_{j,i}[t] > 0$ only when $\tilde{X}_j[t] = (X_j[t]-L_j[t])^+ > 0,$ i.e., MG$_j$ has excess harvested 
energy (i.e. the harvested energy is greater than its demand).

\item[3.] Transfer energy from the macro-grid

In case the energy from the battery and the
energy borrowed from neighboring MGs is insufficient to satisfy the demand, MG$_i$ can borrow $G_i[t]$ units of energy from the marco-grid.
\end{itemize}
The sum of energy drawn from the battery, energy exchange with the neighboring MGs, and the energy borrowed
from the macro-grid must satisfy the residual demand, i.e.,
\begin{align}
B_{i,i}[t]+ \sum_{j \neq i} & B_{j,i}[t] + G_i[t]
= \tilde{L}_i[t] \nonumber \\ & \forall t, \ i = 1,\dots,N. \label{eqn:energyTxfer}
\end{align}

$\bullet$ Now we consider the second case, in which the harvested energy exceeds the energy demand, i.e., if $X_i[t] \geq L_i[t],$ then the MG does the following:
\begin{itemize}
\item[1.] As previously mentioned, it can donate an amount 
$B_{i,j}[t]$ to satisfy the load of MG$_j.$
\item[2.] Store an amount $Y_{i}[t] \leq Y_{\max}$
(where $Y_{\max} < \infty$) in its own battery to be used at a later time.
Accordingly, at each time $t,$ we have
\begin{align}
Y_i[t] + \sum_{j \neq i} B_{i,j}[t] \leq \tilde{X}_i[t] \qquad \forall t, \ i = 1,\dots,N. \label{eqn:InputEnergy}
\end{align}
\end{itemize}

\subsubsection{Energy Storage}
Next, we consider the energy model for the battery at the MGs. At MG$_{i}$, the battery evolves according to the following rule:
\begin{align}
E_i[t+1] = E_i[t]- B_{i,i}[t] +Y_i[t] \ \forall t, \ i = 1,\dots,N \label{eqn:batteryEvol},
\end{align}
where the energy availability constrains the battery at each MG$_{i}$ to satisfy
\begin{align}
B_{i,i}[t] \leq E_i[t]
\ \forall t, \ i = 1,\dots,N \label{eqn:battEnAvailab}.
\end{align}
We also impose a battery discharge constraint during every time slot $t$, namely, 
\begin{align}
B_{i,i}[t] \leq B^s_{\max}
\qquad \ \  \forall t, \ i = 1,\dots,N \label{eqn:battDis},
\end{align}
where $B^s_{\max}$ is the maximum discharge of the battery
per time slot.
Furthermore, the energy storage device has finite capacity of $E_{\max}$ units as follows:
\begin{align}
E_i[t] \leq E_{\max} \ \ \forall t, \ i = 1,\dots,N \label{eqn:battCap}.
\end{align}
Additionally, we make the following practical assumption
on the battery capacity:
\begin{align}
E_{\max} > Y_{\max}+B^s_{\max} \label{eqn:pracconstraint}.
\end{align}
The constraints \eqref{eqn:battEnAvailab} and \eqref{eqn:battDis} can be combined as follows:
we have
\begin{align}
B_{i,i}[t] \leq \min(E_i[t],B^s_{\max}) \ \ \forall t , \ i = 1,2,\dots,N \label{eqn:consolidatedBattDis}.
\end{align}
Similarly the battery input energy constraint
$Y_i[t] \leq Y_{\max}$ and \eqref{eqn:battCap} can be combined as
\begin{align}
Y_{i}[t] \leq \min(E_{\max}-E_i[t],Y_{\max}) \ \ \forall t , \ i = 1,2,\dots,N \label{eqn:consolidatedBattCharge}. 
\end{align}

\subsubsection{Cost Model and Problem Formulation}
We consider that transferring energy
from MG$_i$ to MG$_j$ during the time slot $t$ incurs a cost of $p_{i,j}[t]$ per unit.
Similarly, transferring energy
from the macro-grid to MG$_i$ incurs a cost of $q_i[t]$ per unit.
For simplicity, we assume that the costs $\{p_{i,j}[t],q_i[t] \}$ are i.i.d. across time slots\footnote{Once again, we point out that  the algorithm developed in this work can be generalized to the case when $\{p_{i,j}[t],q_i[t] \}$ are Markovian.}.
Further, we assume that the costs are bounded as $p_{i,j}[t] \leq p_{\max} \ \forall i,j,t$ for some finite $p_{\max},$ and 
$q_{i}[t] \leq q_{\max} \ \forall i,t$ for some finite $q_{\max}.$

The total cost incurred for energy transfer to
MG$_i$ during the time slot $t$ (denoted by $\text{Cost}_{i}[t]$)is given by
\begin{align}
\text{Cost}_{i}[t] = q_i[t] G_i[t] + \sum_{j \neq i} p_{j,i}[t] B_{j,i}[t],
\ i = 1,\dots,N \label{eqn:inputEnConstraint}.
\end{align}
The objective of the controller is to design the system parameters 
in order to minimize the time average cost of energy transfer across the grid subject to the renewable energy generation and battery constraints, stated as follows:
\beqa
&\dsp \min &   \limsup_{T \to \infty} \frac{1}{T}\sum^{T-1}_{t=0} \mathbb{E} \Big{[} \sum^N_{i = 1} \text{Cost}_{i}[t]  \Big{]} \label{eqn:CostMinOpt} \\
& s.t. & B_{i,i}[t]+ \sum_{j \neq i}  B_{j,i}[t] + G_i[t]
= \tilde{L}_i[t], \ \ \forall t, \forall i \nonumber \\
& & Y_i[t] + \sum_{j \neq i} B_{i,j}[t] \leq \tilde{X}_i[t], \qquad \qquad \ \forall t, \forall i \nonumber \\
& & E_i[t+1] = E_i[t]- B_{i,i}[t] +Y_i[t], \ \ \ \ \forall t, \forall i \nonumber \\
& & B_{i,i}[t] \leq \min(E_i[t],B^s_{\max}) \qquad \qquad \ \  \forall t, \forall i \nonumber \\
& & Y_{i}[t] \leq \min(E_{\max}-E_i[t],Y_{\max}) \qquad   \forall t, \forall i \nonumber \\
& & B_{i,j}[t] \leq B^{\text{ex}}_{\max} \qquad \qquad \qquad  \forall t,\forall j \neq i, \forall i. \nonumber
 \eeqa
During each time slot $t,$ the decision variables
are $Y_i[t],B_{i,i}[t],B_{i,j}[t] (\forall j \neq i),G_{i}[t] \ \forall i.$
We denote the optimal value of the cost function $\lim_{T \to \infty} \frac{1}{T}\sum^{T-1}_{t=0} \mathbb{E} \Big{[} \sum^N_{i = 1} \text{Cost}_{i}[t]  \Big{]}$ over all possible control actions
by $f^*_N.$ Note that we have explicitly mentioned 
the subscript $N$ to denote the optimal solution when
$N$ MGs cooperate.

\subsection{Discussion of the System Model}
We now provide some remarks on the system model. 

$\bullet$ Note that in this work, we assume that the excess renewable
energy from MG$_i$ can be exchanged with MG$_j$ only to satisfy
the load of MG$_j$ during the same time slot. In other words,
MG$_j$ cannot use the excess renewable energy of MG$_i$ to charge
its battery. Further, we assume that the energy stored in the
battery of MG$_i$ can only be used to serve its own load during
a future time slot, and there is no energy exchange possible
from the battery. These restrictions are imposed in order to clearly exhibit the trade-off between storing energy for future use and cooperation
in the current time slot (by exchanging energy in the current time slot with neighboring MGs). The framework developed in this work can be easily extended to incorporate all these cases.

$\bullet$ Transferring energy between
different elements of the grid incurs energy losses.
The cost of energy transfer $\{ p_{i,j}[t],q_i[t] \}$ considered in this work can be
interpreted as a price paid for these power losses. 
Typically, the distance between the MGs is much less compared to their distance from the macro-grid. Therefore, in practice, the cost of energy
exchange between the MGs is lower, i.e., $ p_{i,j}[t] < q_i[t].$
Moreover, the MGs can
be connected by a short distance DC power line which incurs
much less energy loss than the AC power line connection
between the MGs and the macro-grid \cite{Larruskain2005}.
The objective function can be viewed 
as minimizing the cost associated with the energy losses.

$\bullet$ A more practical set-up in which the MGs are connected to different buses, and energy exchange performed via a constrained transmission network is not considered this work.
This is done in order to clearly illustrate the trade-offs between cooperation and storage, without the complications
arising from physical power flow constraints.
The more practical case incorporating 
transmission network constraints is left for future investigation. 

\section{An analytical characterization of the optimal cost}
\label{sec:analoptcost}
Before solving the optimization problem in \eqref{eqn:CostMinOpt},
we first provide an analytical 
characterization of the time average cost for certain specific scenarios.

\subsection{Single MG scenario}
First, we consider the analysis only for a single
MG which is equipped with a storage device.
There is no energy to be exchanged, and 
any unsatisfied load must be fulfilled by
using the energy stored in storage device, or by
borrowing energy from the macro-grid.
We make the following simplifications.
\subsubsection{Modeling the energy arrival process}
We assume that
the energy arrivals per time slot are integer random variables.
In addition, we consider the probability mass function (p.m.f.) of the excess energy arrival process\footnote{In this section, we suppress the subscript $i$ since we are considering a single MG.}
i.e., 
\begin{align}
\tilde{X}[t]= X[t]-L
\end{align}
to be given by
\begin{align}
f_{\tilde{X}}(x)  = \begin{cases} M  & \ \text{w.p. \ $a_{M}$}  \\
                      M-1  & \ \text{w.p. \ $a_{M-1}$} \\
                      \vdots &   \\
                      0 & \ \text{w.p. \ $a_0$} \\
                      \vdots & \\
                       -(M-1)  & \ \text{w.p. \ $d_{M-1}$} \\       
                       -M  & \ \text{w.p. \ $d_{M},$} \\       
                    \end{cases}  \label{eqn:genpdf}
\end{align}
where
\begin{align}
d_{M} + d_{M-1} + \dots +a_0+\dots+ a_{M-1}+ a_{M} = 1.
\end{align}
Here, we have assumed that $-M \leq \tilde{X}[t] \leq M.$
Note that the random process $\tilde{X}[t]$ can be used to
model discrete random processes like the
Poisson process.

\subsubsection{Modeling the storage}
As before, let us assume that
the MG has storage capability represented
as a virtual energy queue whose
queue-length at time $t$ is given 
by $E[t],$ and the maximum battery capacity
being $E_{\max}.$
The evolution of the battery can be modeled 
as a random walk which evolves as follows:
\begin{align}
E[t+1] = \min \big{(}\max(E[t]+\tilde{X}[t],0),E_{\max}\big{)}.
\end{align}
Assuming that the arrivals are i.i.d. across time,
it can be verified that the random process
$(E[t], t\geq 0)$ is Markovian.
The state diagram of the Markov chain 
corresponding to the random process $E[t]$
is shown in Figure \ref{fig:MarkovChain}.
We denote $\pi_t[i] = \mathbb{P}(E[t] = i).$
Let us assume that the limiting state distribution function
$\pi(i) = \lim_{t \to \infty} \pi_t(i)$
exists for all $ 0 \leq i \leq E_{\max}$
and denote
$\piv = [\pi(1),\pi(2),\dots,\pi(E_{\max})].$
We denote the transition
matrix of this Markov 
chain by $\Pm$.
The Markov chain corresponding to
the random process $E[t]$ has a limiting 
distribution $\pi(i) = \lim_{t \to \infty} \pi_t(i)$
for all $ 0 \leq i \leq E_{\max},$ independent of
the initial distribution, if the system of equations
\begin{align}
\piv = \piv \Pm \\
\piv {\bf 1} = 1
\end{align}
has a strictly positive solution \cite{gallager1995}.

If the limiting distribution exists, then an analytical expression for the cost incurred in borrowing energy from the macro-grid 
(denoted by $\overline{\text{Cost}}$) can be derived as follows:
In the steady state, when the excess renewable energy
arrival is negative, i.e., $\tilde{X} = i$ for 
$i = -1,\dots,-M$ (which happens with probability
$d_i, i = 1,\dots,M,$ respectively), and
when the battery is in the state
$j$ for $j \leq i,$ (which happens with probability $\pi_j$), the MG has to borrow $i-j$ units of energy from
the marco-grid at the price $q_{\max}$ per unit.
Mathematically we can write,
\begin{align}
\overline{\text{Cost}} = q_{\max} \sum^M_{i = 1} \sum^ i _{j = 0} (i-j)  d_{i} \pi(j).
\end{align}

\begin{figure*}[!htbp]
\begin{center}
\includegraphics[width=0.8\textwidth]{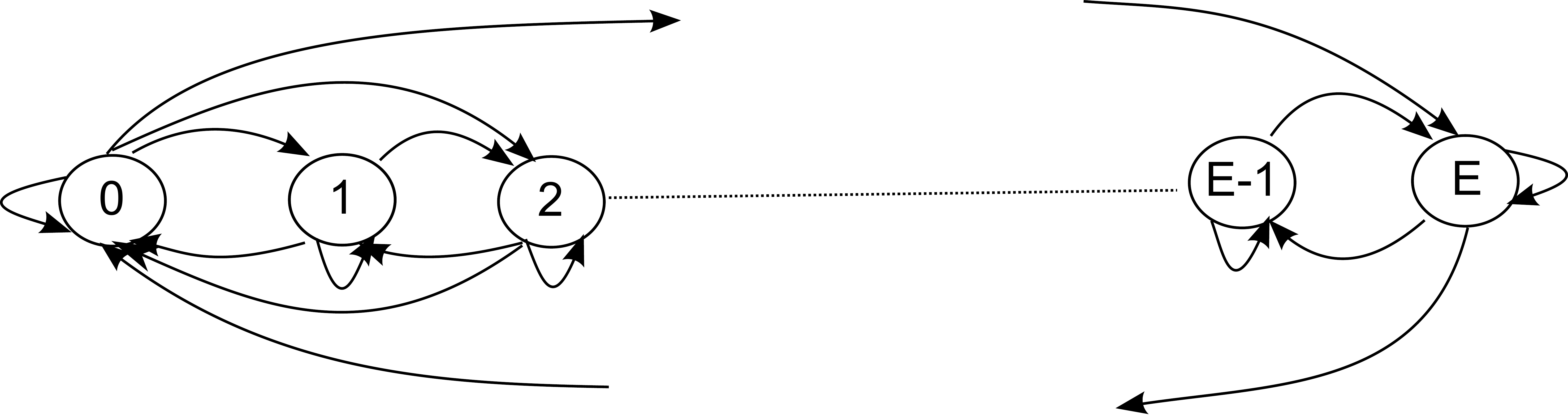}
\caption{State diagram for the Markov chain modeling
for the energy storage model with arrival process given in 
\eqref{eqn:genpdf}.}
\label{fig:MarkovChain}
\end{center}
\end{figure*}

\begin{example}
Consider the special case in which
the the excess energy arrival process has the 
following distribution:
\begin{align}
f_{\tilde{X}}(x) = \begin{cases} -1  & \ \text{w.p. \ $d$}  \\                      
                      0 &   \   \text{w.p. \ $1-a-d$} \\
                      1 &   \   \text{w.p. \ $a$}.                
                   \end{cases} \label{eqn:simplepdf}
\end{align}
The Markov chain associated with the random process $E[t]$ in this case is illustrated in
Figure \ref{fig:MarkovChainOneHop}.
\begin{figure}[!htbp]
\includegraphics[width=0.48\textwidth]{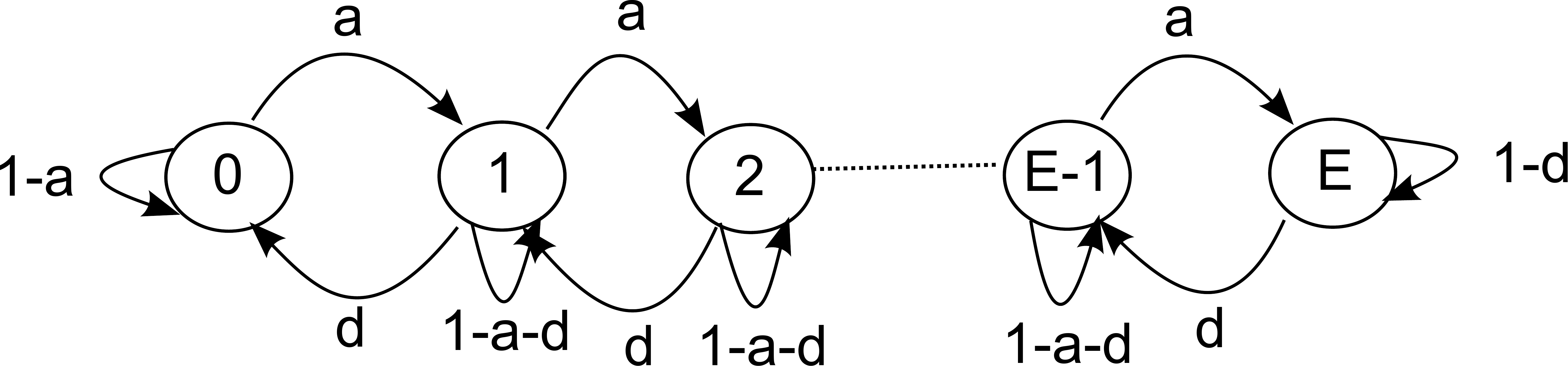}
\caption{State diagram for the Markov chain modeling
of the energy storage with p.m.f. of the arrival process 
as in \eqref{eqn:simplepdf}.}
\label{fig:MarkovChainOneHop}
\end{figure}
In this case, the stationary distribution 
of the Markov chain corresponding to the 
random walk has a simple form given by
\begin{align}
\pi(i) = r^i \LB \frac{1-r}{1-r^{(E_{\max}+1)}}\RB
\end{align}
where $r = a/d.$
Therefore, the cost of borrowing energy
from the marco-grid is given by
\begin{align}
\overline{\text{Cost}} & = q_{\max} \mathbb{P} (\tilde{X} = -1) \pi(0) = q_{\max} d \pi(0) \nonumber \\ & = q_{\max} d  \LB \frac{1-r}{1-r^{(E_{\max}+1)}} \RB \label{eqn:singleMGAnal}.
\end{align}
\end{example}

\subsection{The case of multiple micro-grids}
The characterization of the 
optimal cost in the case of multiple 
MGs is complicated due to the fact the
excess energy arrivals in
the grids are correlated because of the 
ability to share energy among themselves.
Therefore, we provide a closed form analytical characterization
only in the special case case of a completely symmetric two MG set-up. 

Consider a symmetric scenario consisting of 
2 MGs (MG$_1$ and MG$_2$ respectively), such that 
the cost of energy exchanges $p_{1,2} = p_{2,1} = p_{\max}$ 
and $q_{1} = p_{2} = q_{\max}.$
Each MG has an
excess energy arrival process 
whose p.m.f is given as in \eqref{eqn:simplepdf}.
Further, we consider a special case when the energy arrival process is independent across the two MGs.
In order to model the energy transfer 
between the grid, we consider the following 
policy.
Whenever, MG$_1$ produces excess energy and
MG$_2$ has an energy deficit, MG$_1$
transfers its excess energy to MG$_2$
with probability $\alpha \in [0,1].$ Otherwise,
MG$_1$ stores the excess energy
into its battery with a probability 
$1-\alpha.$
Since the system is perfectly symmetric, MG$_2$
does the same in the case when it over produces
and MG$_1$ has an energy deficit. In all other 
cases, there is no requirement to exchange energy between
them. Let us denote the random variable $\tilde{Z}_1$ 
and $\tilde{Z}_2$ representing the effective
excess energy arrival. It is related to
 $\tilde{X}_1$ and $\tilde{X}_2$ as follows:
\begin{itemize}
\item If $\tilde{X}_1 = 1 \ \text{and} \ \tilde{X}_2 = -1  $, then
\begin{align}
\tilde{Z}_1 = \begin{cases} 1  & \ \text{w.p. \ $(1-\alpha)ad$}  \\                      
                      0 &   \   \text{w.p. \ $ \alpha a d $}, 
                   \end{cases} 
\end{align}
and 
\begin{align}
\tilde{Z}_2 = \begin{cases} -1  & \ \text{w.p. \ $(1-\alpha)ad$}  \\                      
                      0 &   \   \text{w.p. \ $ \alpha a d $}. 
                   \end{cases} 
\end{align}
\item If $\tilde{X}_1 = -1 \ \text{and} \ \tilde{X}_2 = 1  $
\begin{align}
\tilde{Z}_1 = \begin{cases} -1  & \ \text{w.p. \ $(1-\alpha)ad$}  \\                      
                      0 &   \   \text{w.p. \ $ \alpha a d $},
                   \end{cases} 
\end{align}
and
\begin{align}
\tilde{Z}_2 = \begin{cases} 1  & \ \text{w.p. \ $(1-\alpha)ad$}  \\                      
                      0 &   \   \text{w.p. \ $ \alpha a d $}. 
                   \end{cases} 
\end{align}
\item In all other cases, $\tilde{Z}_1 = \tilde{X}_1$ and $\tilde{Z}_2 = \tilde{X}_2.$
\end{itemize}
Therefore, the unconditional p.m.f. of 
$\tilde{Z}_1$ and  $\tilde{Z}_2$ can be computed by 
integrating out the other variable.
It is given as
\begin{align}
f_{\tilde{Z}_1}(z) = \begin{cases} -1  & \ \text{w.p. \ $d(1-\alpha a)$}  \\                      
                      0 &   \   \text{w.p. \ $2 \alpha a d + (1-a-d)$} \\
                      1 &   \   \text{w.p. \ $a(1-\alpha d).$}                
                   \end{cases} 
\end{align}
The evolution of the battery can now be modeled 
as a random walk which evolves as 
\begin{align}
E_i[t+1] = \min(\max(E_i[t]+\tilde{Z}_i[t],0),E_{\max}).
\end{align}
The steady state distribution of the Markov chain
corresponding to this random walk is given by
\begin{align}
\pi(j) =  r^j \LB \frac{1-r}{1-r^{(E_{\max}+1)}} \RB,
\end{align}
where
\begin{align}
r = \frac{a(1-\alpha d)}{d(1-\alpha a)}.
\end{align}
Therefore, the cost of energy exchange within the grid has
two components: energy exchanged among the MGs 
and the energy borrowed from the macro-grid. 
Mathematically, this is given as
\begin{align}
\text{Cost}(\alpha) & = 2 \alpha ad p_{\max}+2 d(1-\alpha a) \pi(0) q_{\max}  \label{eqn:Cost2Grid}.
\end{align}
The minimum cost is then obtained by optimizing 
over the choice of $\alpha.$
Let us define
\begin{align}
\alpha^* = \argmin_{\alpha} \text{Cost}(\alpha),
\end{align}
and hence, the minimum cost of energy exchange 
is given by $\text{Cost}(\alpha^*).$
From \eqref{eqn:Cost2Grid}, we can analyze the following two extreme cases, namely, the case with no storage and
the case with infinite storage.
\begin{itemize}
\item $E_{\max} = 0$ - No storage \\
In this case $\pi(0) = 1.$
Therefore, 
\begin{align}
\overline{\text{Cost}} = 2 d q_{\max}+ ad \alpha (- q_{\max} +  p_{\max}),
\end{align}
and the cost is minimized when,
\begin{align}
\alpha = \begin{cases} 1  & \ \text{if \ $p_{\max} < q_{\max}$ }  \\                      
                      0 &   \   \text{else.}                     
                   \end{cases} 
\end{align}
Thus, in the absence of storage, the MGs must always 
share the excess available energy as long as
$p_{\max} < q_{\max}.$ The result is quite intuitive since in the absence of storage, one can always reduce the cost 
by exchanging energy locally between the MGs.
\item $E_{\max} = \infty$ - Infinite storage \\
In this case, $\pi(0) = \frac{d-a}{d(1-\alpha a)}.$
Therefore, 
\begin{align}
\overline{\text{Cost}} = 2 d (d-a) q_{\max}  + 2 ad \alpha.
\end{align}
Thus, the cost is minimized when
$\alpha = 0.$
This implies that in the presence of infinite storage,
any excess energy must always be stored 
rather than exchanging among them.
\end{itemize}

In what follows, we provide a practical algorithm 
to solve the time average cost minimization problem across the grid using the technique of Lyapunov optimization technique. 

\section{Algorithm Design based on Lyapunov Optimization}
\label{sec:LypOpt}
The Lyapunov optimization method provides simple online
solutions based only on the current knowledge of the system
state as opposed to traditional approaches such as dynamic
programming and Markov decision processes which suffer from very high complexity and require a-priori knowledge of the statistics of all the random processes in the system.

But first, we note that the technique of Lyapunov optimization is not directly applicable for solving \eqref{eqn:CostMinOpt}.
This is due to the presence of constraints \eqref{eqn:consolidatedBattDis}
and \eqref{eqn:consolidatedBattCharge}, which have the effect of
coupling the control decisions across time slots.
In order to circumvent this issue, we consider 
an approach similar to \cite{UrgaonkarSigmetric2011}, and
formulate a slightly modified version of this problem, 
stated as follows:
\beqa
&\dsp \min &   \lim_{T \to \infty} \frac{1}{T}\sum^{T-1}_{t=0} \mathbb{E} \Big{[} \sum^N_{i = 1} \text{Cost}_{i}[t]  \Big{]} \label{eqn:CostMinRelaxed} \\
& s.t. & \lim_{T \to \infty} \frac{1}{T}\sum^{T-1}_{t=0} \mathbb{E}[Y_i[t]] \leq \lim_{T \to \infty} \frac{1}{T}\sum^{T-1}_{t=0} \mathbb{E}[{B}_{i,i}[t]] \ \forall i \nonumber  \\
& & B_{i,i}[t]+ \sum_{j \neq i}  B_{j,i}[t] + G_i[t]
= \tilde{L}_i[t], \ \ \forall t, \forall i \nonumber \\
& & Y_i[t] + \sum_{j \neq i} B_{i,j}[t] \leq \tilde{X}_i[t], \qquad \qquad \ \forall t, \forall i \nonumber \\
& & B_{i,i}[t] \leq B^{\text{s}}_{\max}; B_{i,j}[t] \leq B^{\text{ex}}_{\max} \ \  \forall t,\forall j \neq i, \forall i, \nonumber \\
& & Y_i[t] \leq Y_{\max}  \forall t,\forall i \nonumber.
 \eeqa
Note that in \eqref{eqn:CostMinRelaxed}, all the constraints 
associated with the battery are relaxed, and a constraints
of the form $\lim_{T \to \infty} \frac{1}{T}\sum^{T-1}_{t=0} \mathbb{E}[Y_i[t]] \leq \lim_{T \to \infty} \frac{1}{T}\sum^{T-1}_{t=0} \mathbb{E}[{B}_{i,i}[t]] \ \forall i$ are added. Effectively, this constraint represents the 
condition for stability of the virtual energy queue associated with the battery. We will henceforth address this problem as the \emph{relaxed problem}. Let us denote by $g^*_N,$ the optimal
value of the cost function of the \emph{relaxed problem}
over all possible control decisions.
First, it is easy to see that $g^*_N \leq f^*_N,$ i.e. the solution to the relaxed problem 
acts as a lower bound on the original problem (since the relaxed problem has fewer constraints compared to the original problem and any feasible solution of the original problem is feasible for the relaxed problem as well). 

We now focus on solving the \emph{relaxed problem}.
It can be shown that the optimal solution to the
\emph{relaxed problem} can be obtained by the
method of stationary randomize policy, stated 
in the following theorem.
\begin{theorem}
\label{thm:StatRandom}
There exists a stationary and randomized policy
$\Pi$ that achieves
\begin{align}
\mathbb{E} \Big{[} \sum^N_{i = 1} \text{Cost}^{\Pi}_{i}[t]  \Big{]} = g^*_N \qquad \forall t,
\end{align}
and satisfies the constraints:
\begin{align}
&  \mathbb{E}[Y^{\Pi}_i[t]] \leq  \mathbb{E}[{B}^{\Pi}_{i,i}[t]] \ \ \forall i,\forall t \nonumber  \\
&  \mathbb{E} [B^{\Pi}_{i,i}[t]]+ \sum_{j \neq i} \mathbb{E} [ B^{\Pi}_{j,i}[t]] + \mathbb{E} [G^{\Pi}_i[t]]
= \tilde{L}_i[t], \ \ \forall t, \forall i \nonumber \\
&  \mathbb{E} [ Y^{\Pi}_i[t]] + \sum_{j \neq i} \mathbb{E} [ B^{\Pi}_{i,j}[t]] \leq \mathbb{E} [\tilde{X}_i[t]], \qquad \qquad \ \forall t, \forall i \nonumber \\
& \mathbb{E}[ B^{\Pi}_{i,i}[t]] \leq B^{\text{s}}_{\max}; \ 
\mathbb{E} [B^{\Pi}_{i,j}[t]] \leq B^{\text{ex}}_{\max}
\qquad  \forall t,\forall j \neq i, \forall i, \nonumber \\
&  \mathbb{E}[ Y^{\Pi}_i[t]] \leq Y_{\max}  \forall t,\forall i \nonumber.  \nonumber
\end{align}
\end{theorem}
The existence of such a policy can be proved 
by using the Caratheodory theorem,
similar to the arguments in \cite{NeelyBook} and omitted
here for brevity.
Note that due to the high dimensionality, it is not practical 
to solve the problem using the method of stationary randomized policies.

In what follows, we apply Lyapunov optimization 
to solve the \emph{relaxed problem} \eqref{eqn:CostMinRelaxed}. Further, we will
show that the solution developed for \eqref{eqn:CostMinRelaxed} by our method also satisfies
all the constraints associated with battery, hence making it applicable for solving the original problem \eqref{eqn:CostMinOpt}.

We proceed by considering the Lyapunov function associated with the virtual energy queues defined as follows: 
\begin{align}
\Psi[t] = \frac{1}{2} \sum_i (E_i[t]-\theta)^2
\end{align}
where $\theta$ is a perturbation
which is given by $\theta = B^s_{\max}+Vq_{\max}.$ 
The exact rationale behind the choice of the value of
$\theta$ will be specified later when we analyze
the algorithm performance.

We will now examine the Lyapunov drift 
which represents the expected change
in the Lyapunov function from one time slot to the other, 
which is defined as 
\begin{align}
\Delta [t] & = \mathbb{E} \LSB \Psi[t+1]-\Psi[t] \big{|} \Em[t] \RSB,
\end{align}
where the expectation is with respect to the random processes 
associated with the system, given the energy queue-length values 
$\Em[t] = [E_1[t],\dots,E_N[t]].$
Using the equation for evolution of the virtual energy
queue associated with the battery in \eqref{eqn:batteryEvol},
and some standard manipulations, it can be shown that
the Lyapunov drift can be bounded as
\begin{align}
\Delta[t] \leq C - \mathbb{E} \Big{[}\sum_i  (E_i[t]-\theta) (B_{i,i}[t]-Y_i[t]) \big{|} \Em[t] \Big{]} \label{eqn:LypUpBound}
\end{align}
where $C < \infty$ is a constant.
For completeness, we provide the proof of this
step in Appendix A.

We will henceforth denote $\tilde{E}_i[t] = E_i[t]-\theta.$
Adding the performance metric $V \mathbb{E} \LSB (\sum_{i,j} p_{i,j} B_{i,j}[t] + \sum_i q_i[t] G_i[t]) | \Em[t]\RSB $ (where $V$ is another control parameter which will be specified later) to both the sides and denoting $$\Delta_V [t] =  \Delta [t] + V \mathbb{E} \big{[} \sum_{i,j} p_{i,j} B_{i,j}[t] + \sum_i q_i[t] G_i[t] | \Em[t] \big{]},$$ we have
\begin{align}
&  \Delta_V[t]  \leq C - \mathbb{E} \Big{[} \sum_i \tilde{E}_i[t] \big{(} B_{i,i}[t]-Y_i[t]) \nonumber \\ & - V \big{(}\sum_{i,j} p_{i,j} B_{i,j}[t]
+ \sum_i q_i[t] G_i[t] \big{)} \big{|} \Em[t] \Big{]} \label{eqn:ModLyp}.
\end{align}
Using \eqref{eqn:energyTxfer}, we have
$G_i[t] = \tilde{L}_i[t] - B_{i,i}[t] - \sum_{j \neq i} B_{j,i}[t].$ Substituting for $G_i[t]$ in the 
right hand side \eqref{eqn:ModLyp}, we have,
\begin{align}
 \Delta_V[t]  & \leq C + \mathbb{E} \Big{[}- \sum_i \tilde{E}_i[t] \big{(} B_{i,i}[t]-Y_i[t]) \nonumber  \\ & + V \big{(} \sum_{i,j} p_{i,j} B_{i,j}[t]  \\ & + \sum_i q_i[t] (\tilde{L}_i[t] - B_{i,i} [t] - \sum_{j \neq i} B_{j,i} [t])  \big{)} \Big{|} \Em[t]\Big{]}\nonumber \\
& = C + \mathbb{E} \Big{[} V \sum_i q_i[t] \tilde{L}_i[t] + \sum_i \tilde{E}_i[t] Y_i[t] \nonumber \\ & \qquad \qquad + \sum_i  \sum_{j \neq i} V B_{i,j} [t] (p_{i,j}[t] -q_j[t]) \big{]} \nonumber \\ & \qquad \qquad - \sum_i B_{i,i}[t] (\tilde{E}_i[t]+V q_i[t]) \Big{|} \Em[t] \Big{]} \label{eqn:LypBound}.
\end{align}
From the theory of Lyapunov optimization (drift-plus penalty method), the control actions are chosen during each time slot to
minimize the bound on the modified Lyapunov drift 
function (on the right hand side of \eqref{eqn:LypBound})
\cite{NeelyBook}.
Before we proceed, we provide the main
intuition behind the solving the \emph{relaxed problem}
using this approach.
Notice that the \emph{relaxed problem} can viewed as minimizing
the time average cost of energy exchange in the grid while maintaining the stability of the virtual energy queue (battery).
The modified Lyapunov drift has two components, the Lyapunov
drift term $\Delta [t],$ and $V \times \text{Cost}[t]$ term. Intuitively, minimizing the Lyapunov drift term 
alone pushes the queue-length of the virtual energy queue to 
a lower value. 
The second metric $V \times \text{Cost}[t]$ can be viewed 
as the penalty term, with the parameter $V$ representing 
the trade-off between minimizing the queue-length drift
and minimizing the penalty function. Greater
value of $V$ represents greater priority 
to minimizing the cost metric at the expense of
greater size of the virtual energy queue and vice versa.
This is indeed the rationale behind 
minimizing the modified Lyapunov drift $\Delta_V[t]$
during each time slot.

The control algorithm using the aforementioned rule can be
described as follows. During each time slot 
$t,$ one must choose the control decisions as 
a solution to the following linear programming
problem (obtained by minimzing the right hand side of \eqref{eqn:LypBound}) :
\beqa
&\dsp \min_{Y_i,B_{i,i},B_{i,j}} &  \sum_i \tilde{E}_i[t] Y_i + V \sum_{j \neq i} B_{i,j} (p_{i,j}[t] -q_j[t]) \nonumber \\ & & \qquad \qquad -  \sum_i B_{i,i} (\tilde{E}_i[t]+V q_i[t]) \label{eqn:OPT_basic_NUM} \\
& s.t. & Y_i + \sum_{j \neq i} B_{i,j} \leq \tilde{X}_i[t] \nonumber \\
& & B_{i,i} + \sum_{j \neq i} B_{j,i} \leq \tilde{L}_i[t] \nonumber \\
& & 0 \leq Y_i \leq Y_{\text{max}}, \ 0 \leq B_{i,i} \leq B^{\text{s}}_{\max} \ \forall i, \nonumber \\ & &  0 \leq B_{i,j} \leq 
B^{\text{ex}}_{\max} \ \forall j \neq i, \forall i \nonumber.
 \eeqa
Let us denote the solution corresponding to
\eqref{eqn:OPT_basic_NUM} as $Y^*_i[t],B^*_{i,i}[t]$ 
and $B^*_{i,j}[t].$
The value of $G^*_i[t]$ is then given by
\begin{align}
G^*_i[t] = (L_i[t] - B^*_{i,i}[t] - \sum_{j \neq i} B^*_{j,i}[t])^+,
\end{align} 
where $(x)^+ = \max(x,0).$
\subsection{Algorithm Performance Analysis}
We will now analyze the performance of the algorithm described in the previous section.

\begin{lemma}
\label{lem:Vbound}
By choosing the parameters $V$ and $\theta$  as
\begin{align}
& 0 < V \leq \frac{E_{\max} - (Y_{\max}+B^s_{\max})}{q_{\max}} \label{eqn:Vbound} \\
& \theta = B^s_{\max} + Vq_{\max} \label{eqn:thetaVal}
\end{align}
the following hold true:
\begin{itemize}
\item[1.] If $E_i [t] > E_{\max}-Y_{\max},$ then $Y^*_{i}[t] = 0.$
\item[2.] If $E_i [t] < B^s_{\max},$ then $B^*_{i,i}[t] = 0.$
\end{itemize}
\end{lemma}
\begin{proof}
See Appendix B.
\end{proof}
From the result of Lemma \ref{lem:Vbound}, it can be seen that the battery charging decisions are non-zero only when $E_i [t] < E_{\max}-Y_{\max}$ 
and the discharge decisions are non-zero only when $E_i [t] > B^s_{\max}.$ 
As a consequence of this lemma, 
it can be verified that the algorithm developed satisfies
all the constraints associated with the battery 
in the original problem, i.e., constraints \eqref{eqn:consolidatedBattDis} and \eqref{eqn:consolidatedBattCharge}.
Therefore, the algorithm developed is a feasible algorithm
for the original optimization problem, and it
can be applied to the original system under consideration
(with finite battery capacity constraints). This also 
justifies the choice of the perturbation parameter $\theta$
in our design.

Let us denote the cost function
$f_N[t] = \sum^N_{i = 1} \text{Cost}_i[t].$ We will now provide the main result related to the performance of our algorithm.
\begin{theorem}
\label{thm:PerBd}
For the algorithm developed in the previous subsection, 
the virtual energy queue-lengths can be bounded as follows:
\begin{align}
0 \leq E_i [t] \leq E_{\max} \qquad \forall t, i = 1,\dots,N,
\end{align}
and the time average cost function
achieved by this algorithm satisfies
\begin{align}
\limsup_{T \to \infty} \frac{1}{T}\sum^{T-1}_{t = 0}\mathbb{E}[{f_N[t]}] \leq g^*_N + \frac{\tilde{B}}{V} \label{eqn:performance}
\end{align}
where $\tilde{B} < \infty$ is a constant.
\end{theorem}
\begin{proof}
See Appendix C.
\end{proof}
Theorem \ref{thm:PerBd} implies that 
performance of our algorithm can be made arbitrarily close to the optimal value 
by increasing the value of parameter $V.$ 
However, this comes at the cost of increasing the battery
capacity $E_{\max}$  (due to the bound on the 
value of $V$ in \eqref{eqn:Vbound} , given $E_{\max}$).
Also, note that since $g^*_{N} \leq f^*_{N},$ the performance bounds of \eqref{eqn:performance} hold with respect $f^*_{N}$ as well (and are infact tighter).

\section{Numerical Results}
\label{sec:numresult}
In this section, we present some numerical results to
examine the analytical modeling of the cost function and the Lyapunov optimization based algorithm described in the previous section.

First, we plot the optimal time average cost as a function of the 
storage size for the single MG
case based on the analytical expression of \eqref{eqn:singleMGAnal}, and
by running the Lyapunov optimization based
algorithm for $T = 5000$ iterations in Figure 
\ref{fig:AnalyticalVsNumerical}. 
The p.m.f of the excess energy arrival process 
considered in the simulations is provided in \eqref{eqn:simplepdf}, with $d = 0.5, a = 0.2.$
We observe that 
there is a good match between the two curves and the 
time average cost decreases with an increase in storage 
capacity.

\begin{figure}[!htbp]
\begin{center}
\includegraphics[width=0.45\textwidth]{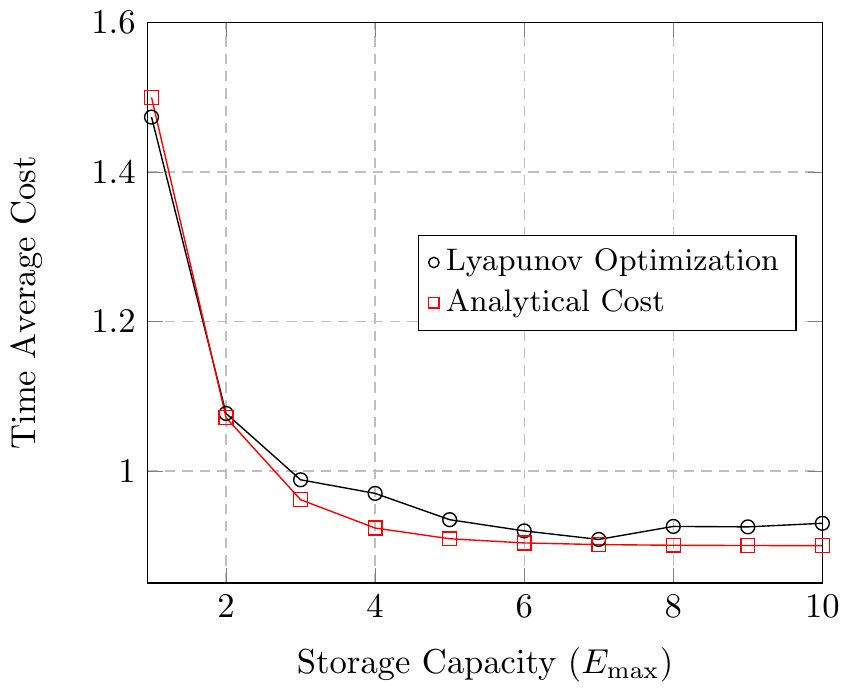}
\caption{Time average cost versus storage capacity based on analytical result and Lyapunov optimization based method.}
\label{fig:AnalyticalVsNumerical}
\end{center}
\end{figure}

Next, we consider the normalized time average 
cost for different combinations of the number
of cooperating MGs and the storage capacity.
We consider the following setting.
We assume that the actual harvested renewable
energy consists of two components, namely
the predicted component which is deterministic, and a prediction
error which is random i.e., $X_i[t] = \hat{X}_i[t]+ W_i[t]$. For the sake of simplicity, we assume
that the predicted component $\hat{X}_i[t]$ perfectly matches the aggregate load $L_i[t].$
We model the excess renewable energy production
$\tilde{X}_i[t] = X_i[t]-L_i[t]$
(also the prediction error $W_i[t]$ in this case) by the truncated normal distribution in accordance with some previous studies in this field \cite{MakarovLoutan2009}.
Therefore, excess renewable energy production
is a zero mean random process with its distribution 
given by $\tilde{X}_i[t] \sim \mathcal{N}(0,\nu^2),$
where $\nu$ is the standard deviation.
Note that following this model of energy production 
(i.e., $\hat{X}_i[t] = L_i[t]$ and mean value of 
$\tilde{X}_i[t] = 0$), we ensure that
the aggregate energy production of each MG matches the
aggregate load, and the only reason to store/exchange energy
is to compensate for the fluctuations in renewable energy
(thus capturing the most essential aspect of the problem
considered in this work).
In our numerical results, we choose
the value of $L_i[t] = L_i = 10 \ \text{MW}, \ \forall t,i$
and the variance $\nu = 3 \ \text{MW}.$

In order to construct a wind farm of MGs, we choose a square
grid of dimensions $10 \times 10$ km. The positions of 
the micro-grids are chosen by generating uniform random 
numbers within this grid. We consider that the macro-grid
is located at the coordinate $(20,20)$ km.
A random snapshot the power grid with $N = 5$ MGs is
shown in Figure \ref{fig:GridView}.

\begin{figure}[!htbp]
\begin{center}
\includegraphics[width=0.40\textwidth]{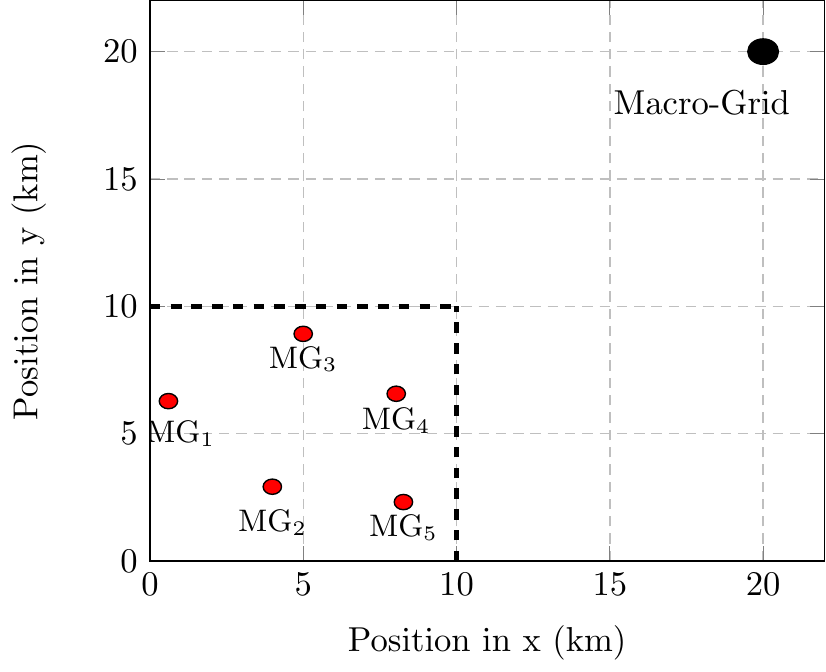}
\caption{A snapshot of the power grid with 5 MGs and the macro-grid. The red dots indicate
the position of the MGs and the black dot indicates the macro-grid. The dotted lines represent the boundary of the wind farm.}
\label{fig:GridView}
\end{center}
\end{figure}

We assume that the cost of transporting energy
among different elements of the grid is directly proportional
to the distance between them.
Therefore, $p_{i,j}[t] = \beta d_{i,j}, \ \forall t,$
units
(where $d_{i,j}$ is the distance between MG$_i$ and
MG$_j$ in km) and $q_{i}[t] = \beta D_{i}, \ \forall t,$
(where $D_{i}$ is the distance between MG$_i$ and the
macro-grid in km) and $\beta$ being the proportionality constant.
In our numerical results, we use $\beta = 1.$
For each random snapshot the power grid, we run the simulation for $T = 5000$ time slots. 
Let us consider the time average cost incurred 
\begin{align*}
\bar{C}_N & = \frac{1}{T}\sum^{T-1}_{t = 0}\sum^N_{i = 1} C_i[t] \\ &=
\frac{1}{T} \sum^N_{i = 1} \big{(}\sum_{j \neq i} \beta d_{i,j}  \sum^{T-1}_{t = 0} B_{i,j}[t] + q_i D_i \sum^{T-1}_{t = 0} G_i[t]\big{)}
\end{align*}
in transferring energy across the grid, where $N$ is the number of cooperating MGs.
In order to average out the position of the elements
of the grid, we generate $100$ random snapshots
of the power grid and obtain an average of the normalized cost
for these $100$ positions.
We plot the normalized cost incurred $\frac{\bar{C}_N}{N}$
as a function of $N,$ for different values of storage capacity
($2$, $5$, $10$, $20$ and $50$ MWh) in Figure \ref{fig:CostVsCoopMG_Storage}. 
In each of these cases, we choose $B^s_{\max}$ and $Y_{\max}$
to satisfy the constraint \eqref{eqn:pracconstraint}.
Specifically, we choose $(B^s_{\max},Y_{\max}) = (0.5,0.5), (1,1), (2,2), (5,5), (10,10)$ MW in the five cases
of storage capacity ($2$, $5$, $10$, $20$ and $50$ MWh) respectively. The value of $B^{\text{ex}}_{\max}$ chosen in $10$ MW.

\begin{figure}[!htbp]
\begin{center}
\includegraphics[width=0.45\textwidth]{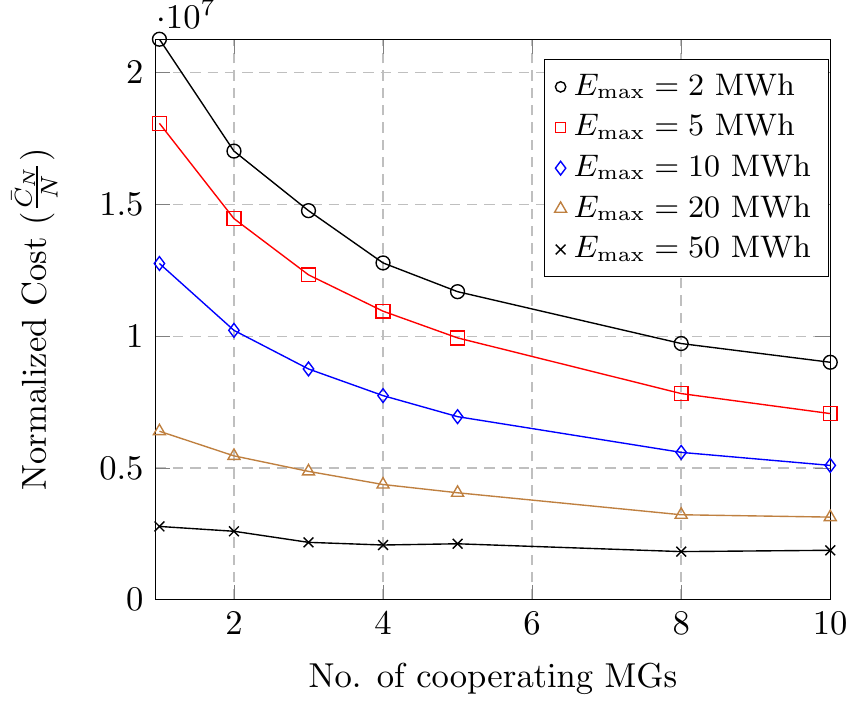}
\caption{Normalized cost versus the number of cooperating MGs for different values of storage capacity.}
\label{fig:CostVsCoopMG_Storage}
\end{center}
\end{figure}
The following observations can be made.
\begin{itemize}
\item For a given storage capacity, the cost of energy exchange decreases with an increase in the number of cooperating MGs.  
This is due to the fact that greater number of cooperating MGs leads to a greater diversity in energy production and hence greater possibility of sharing energy among the MGs (hence reducing the need for borrowing energy from the macro-grid).
\item The decrease in the normalized cost (as a function of
the number of cooperating MGs) is greater for lower values of 
storage capacity. For higher values of storage capacity, the 
normalized cost does not reduce with increasing $N$. This is due to the fact that with greater storage capacity, the MGs are able to store any excess harvested energy (during the time slots when the harvested energy is greater than the aggregate load) and use it during the time slots when the harvested energy is deficient, thereby, eliminating the need for energy cooperation.
\item Further, most of the reduction in the time average cost
is achieved by cooperation among only a few neighboring MGs.
The incremental gain obtained by cooperation among large number of MGs is not significant. 
\end{itemize}

Next, we try to examine the following question:
for a given number of cooperating MGs, what is the storage
capacity needed per MG to eliminate the need for borrowing 
energy from the marco-grid? In order to do so, we 
consider a hypothetical scenario in which
$p_{i,j}[t] = p_{i,j} = \beta $ unit $\forall i,j,t$ and $q_i[t] = q_i = 3 \beta$ units $\forall i,t$ (both being independent of
the distance between elements of the grid).
Once again, we choose $L_i[t] = L_i = 10 \ \text{MW}, \ \forall t,i$ and $\nu = 3 \ \text{MW}.$
We look for the combination of 
the storage and number of cooperating MGs that
yields a normalized time average cost 
\begin{align*}
\frac{\bar{C}}{N} = \beta \times 10^7 \ \text{units}
\end{align*}
as the bench mark (The choice of $\beta \times 10^7$ units comes from the fact that 
$p_{i,j}[t] = \beta \ \text{units}$ and 
$10^7$ corresponds to the $10$ MW demand per time slot).
 We plot the optimal value of $N$ 
and $E_{\max}$ required to make the normalized cost below 1 unit
in Figure \ref{fig:CoopVsStorage}.
In our numerical results, we choose $\beta = 1.$
It can be seen that for a given number of cooperating MGs,
there exists an optimal storage capacity requirement to eliminate
the need for borrowing energy from the macro-grid.
It is evident that as the number of cooperating MGs increases,
the optimal storage capacity requirement reduces.

\begin{figure}[!htbp]
\begin{center}
\includegraphics[width=0.45\textwidth]{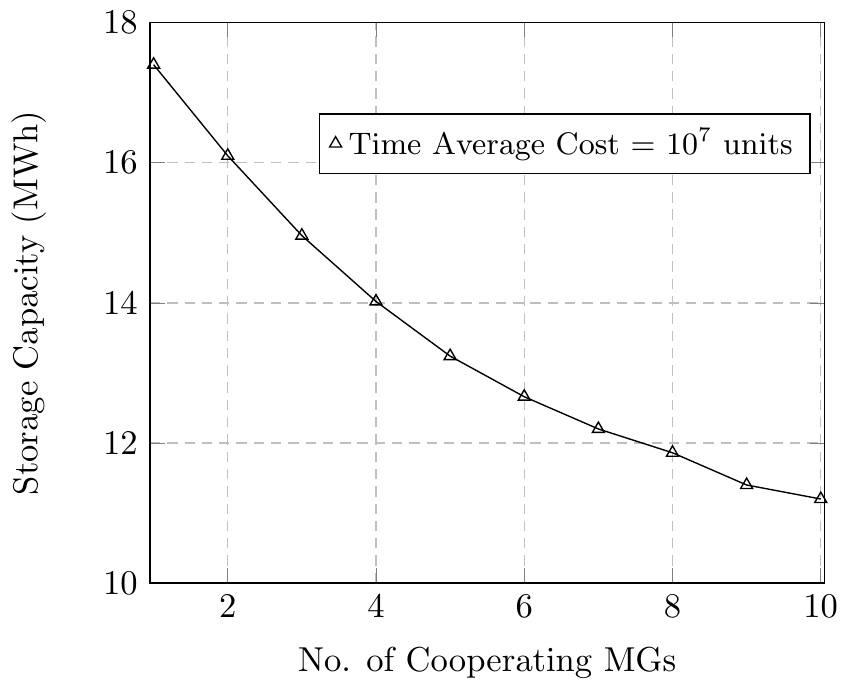}
\caption{Storage Capacity Vs Number of Cooperating MGs to achieve a time average cost of $10^7$ units.}
\label{fig:CoopVsStorage}
\end{center}
\end{figure}

A more practical case in which the algorithm
is implemented on the renewable energy data provided by
National Renewable Energy Laboratory (NREL) of the United
States is presented in \cite{LakshQuekPoor2014} and
similar results have been obtained.

\section{Conclusions and Future Work}
\label{sec:conclusion}
In this work, we explored the benefits of energy storage and cooperation among interconnected MGs as a means to combat the uncertainty in harvesting renewable energy. We modeled 
the set-up as an optimization problem to minimize the 
cost of energy exchange among the grid for a given storage 
capacity at the MGs.
First, we provided an analytical expression for the time average
cost of energy exchange in the presence of storage
by analyzing the steady state of the system 
 for some special cases of interest.
We then developed an algorithm based on Lyapunov optimization
to solve this problem, and provided performance analysis
of this algorithm.
Our results show that in the presence of limited storage devices, the grid can benefit greatly by cooperating even among only a few distributed sources.
However, in the presence
of large storage, cooperation does not yield much benefit.
Our solution can be useful for the power
grid designer in terms of choosing the optimal
combination of storage size and cooperation in order to meet a specific cost criterion. 
Since this work is a first step towards exploring the trade-offs between cooperation and storage, we ignored the transmission network constraints throughout. The analysis of how these constraints impact the trade-off between cooperation and storage is left for future investigation.

\section*{Appendix A : Proof of \eqref{eqn:LypUpBound}}
First note that
\begin{align*}
E_i[t+1]-\theta = E_i[t]-\theta -B_{i,i}[t]+ Y_i[t]
\end{align*}
Squaring both sides of the equation, we obtain,
\begin{align}
(E_i[t+1]-\theta)^2 & = (E_i[t]-\theta+Y_i[t]-B_{i,i}[t])^2 \nonumber \\
& = (E_i[t]-\theta)^2+(B_{i,i}[t]-Y_i[t])^2 \nonumber \\ & \qquad -2 (E_i[t]-\theta)
(B_{i,i}[t]-Y_i[t]) \label{eqn:papa20}
\end{align}
The term $(B_{i,i}[t]-Y_i[t])^2 \leq Y^2_{\max}+B^2_{\max} \defines C.$
Using this bound and rearranging \eqref{eqn:papa20}, we obtain
\begin{align}
(E_i[t+1]-\theta)^2 & - (E_i[t]-\theta)^2 \leq 
C \nonumber \\ & -2 (E_i[t]-\theta)
(B_{i,i}[t]-Y_i[t]).
\end{align}
Summing it over $i = 1,\dots,N$ and taking the conditional expectation on both sides given
$\Em[t]$, we obtain the bound in \eqref{eqn:LypUpBound}.

\section*{Appendix B : Proof of Lemma \ref{lem:Vbound}}
Let us first focus on statement 1.
First, note that since the objective is
to minimize \eqref{eqn:OPT_basic_NUM}, it can be 
easily inferred that 
$Y^*_i[t] = 0 $ when $\tilde{E}_i[t] > 0,$
i.e., $E_i[t] - Vq_{\max}+B^s_{\max} > 0.$
This implies that $Y^*_i[t] = 0 $ when
\begin{align}
 E_i[t] > Vq_{\max}+B^s_{\max} \label{eqn:jaja111}.
\end{align}
Using the bound on the value of $V$ from 
\eqref{eqn:Vbound}, in equation\eqref{eqn:jaja111}, we can conclude that
$Y^*_i[t] = 0 $ when
\begin{align}
 E_i[t] & > E_{\max} - (B^s_{\max}+Y_{\max})+B^s_{\max} \nonumber \\
  & = E_{\max} - Y_{\max}.
\end{align}
Let us now turn to statement 2.
Once again, from the objective function in
\eqref{eqn:OPT_basic_NUM}, it is clear that
$B^*_{i,i}[t] = 0$ if $\tilde{E}_i[t] + Vq_i[t] < 0.$
Substituting for $\tilde{E}_i[t],$ we have,
\begin{align}
& E_i[t] -( Vq_{\max}+B^s_{\max}) + Vq_i[t] < 0 \nonumber \\
& E_i[t]- B^s_{\max} < 0, 
\end{align}
where the last step follows since $q_i[t] \leq q_{\max}.$
Therefore, $B^*_{i,i}[t] = 0$ when $E_i[t]- B^s_{\max} < 0.$

\section*{Appendix C: Proof of Theorem \ref{thm:PerBd}}
We will first prove the bound on the battery size
$0 \leq E[t] \leq E_{\max}.$ We use the analysis
similar to \cite{UrgaonkarSigmetric2011} to obtain 
a bound on the battery size.
We consider four cases.
\begin{itemize}
\item Case 1 : $Vq_{\max}+B^s_{\max} \leq E_i[t] \leq E_{\max}.$
In this case, we have that $Y^*_i[t] = 0.$ Therefore, 
$E_i[t+1] \leq E_i[t] \leq E_{\max}.$
\item Case 2 : $  E_i[t] < Vq_{\max}+B^s_{\max}.$
In this case, we have that $Y^*_i[t] \leq Y_{\max}.$
Therefore, $E[t+1] \leq Vq_{\max}+B^s_{\max}+Y_{\max} \leq E_{\max}, $ where the last inequality follows 
from the range of values of $V$ as considered in \eqref{eqn:Vbound}.
\item Case 3 : $0 \leq E[t] \leq B^s_{\max}.$
In this case, $B^*_{i,i}[t] = 0.$ Therefore, 
$E[t+1] \geq E[t] \geq 0.$
\item Case 4 : $E[t] > B^s_{\max}.$
In this case, $B^*_{i,i}[t] \leq  B^s_{\max}$
and therefore, $E[t] \geq 0.$
\end{itemize}

We next proceed to prove the result on
time average performance of the algorithm.
Consider the bound on the Lyapunov drift function of
\eqref{eqn:LypBound}.
It is clear that the control actions chosen
according to the solution of \eqref{eqn:OPT_basic_NUM}
minimize the bound on the Lyapunov
function over all possible control actions.
Comparing it with the control action
chosen according to any stationary and randomized
policy (which we will denote with the
superscript $\Pi$), we have,
\begin{align}
 \Delta_V[t]  
& \leq C + \mathbb{E} \Big{[} V \sum_i q_i[t] \tilde{L}_i[t] + \sum_i \tilde{E}_i[t] Y^{*}_i[t] \nonumber \\ & \qquad + \sum_{i ,j \neq i} V B^{*}_{i,j} [t] (p_{i,j}[t] -q_j[t])  \nonumber \\ & \qquad - \sum_i B^{*}_{i,i}[t] (\tilde{E}_i[t]+V q_i[t]) \big{|} \Em[t] \Big{]}  \\
& \leq C + \mathbb{E} \Big{[} V \sum_i q_i[t] \tilde{L}_i[t] + \sum_i \tilde{E}_i[t] Y^{\Pi}_i[t] \nonumber \\ & \qquad + \sum_{i,j \neq i} V B^{\Pi}_{i,j} [t] (p_{i,j}[t] -q_j[t]) \nonumber \\ & \qquad - \sum_i B^{\Pi}_{i,i}[t] (\tilde{E}_i[t]+V q_i[t]) \big{|} \Em[t] \Big{]}.
\end{align}
Rearranging, we obtain
\begin{align}
 & \Delta_V[t]  
 \leq C + \mathbb{E} \Big{[} V \sum_i q_i[t] (\tilde{L}_i[t] - B^{\Pi}_{i,i}[t]-\sum_{j \neq i}B^{\Pi}_{j,i} [t]) \nonumber \\ & +\sum_i \tilde{E}_i[t]( Y^{\Pi}_i[t]-B^{\Pi}_{i,i}[t]) + \sum_{i,j \neq i}   V p_{i,j}[t] B^{\Pi}_{i,j} [t]  \big{|}
 \Em[t] \Big{]} \nonumber \\
 & = C + \mathbb{E} \Big{[} V \sum_i q_i[t] G^{\Pi}_i[t]\nonumber \\ & +\sum_i \tilde{E}_i[t]( Y^{\Pi}_i[t]-B^{\Pi}_{i,i}[t]) + \sum_{i,j \neq i}   V p_{i,j}[t] B^{\Pi}_{i,j} [t]  \big{|}
 \Em[t] \Big{]} \label{eqn:here1321},
\end{align}
where the last step we have used the fact that
$G^{\Pi}_i[t] = \tilde{L}_i[t] - B^{\Pi}_{i,i}[t]+\sum_{j \neq i}B^{\Pi}_{j,i} [t].$
In particular, let us consider the 
stationary and randomized
policy $\Pi$ from Theorem
\ref{thm:StatRandom},
which achieves the cost $g^*_N.$ 
Using this policy in \eqref{eqn:here1321}, we obtain
\begin{align}
 & \Delta_V[t]  \leq C + V g^*_N.
\end{align}
Taking the expectation on both the sides 
and summing from $t = 0,\dots,T-1$
, normalizing by $T$ and taking the limit,
it can be shown that 
\begin{align}
\limsup_{T \to \infty} \frac{1}{T} \sum^{T-1}_{t = 0} \mathbb{E}[f_N[t]] \leq g^*_N+\frac{\tilde{B}}{V}.
\end{align}

\bibliographystyle{IEEEtran}
\bibliography{IEEEabrv,bibliography}

\begin{IEEEbiography}
[{\includegraphics[width=1in,height=1.25in,clip,keepaspectratio]{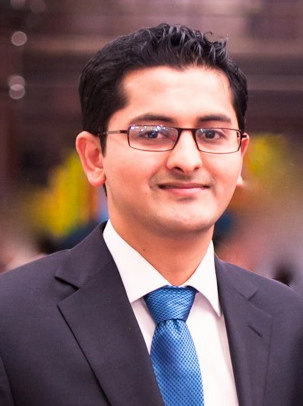}}]
{Subhash Lakshminarayana}
(S'07-M'12)
received his M.S. degree in Electrical
and Computer Engineering from The Ohio State
University in 2009, and his Ph.D. from the
Alcatel Lucent Chair on Flexible Radio
and the Department of Telecommunications at SUPELEC, France in 2012.
Currently he is with The Singapore University of Technology
and Design (SUTD). 

Dr. Lakshminarayana has held a visiting research appointment
at Princeton University from Aug-Dec 2013. He has also been has been a student researcher at the Indian Institute of Science, Bangalore during 2007. He has served as a member of Technical Program Committee for IEEE PIMRC in 2014, IEEE VTC in 2014, and IEEE WCNC in 2014.
His research interests broadly spans wireless communication and signal processing with emphasis on small cell networks (SCNs), cross-layer design wireless networks, MIMO systems, stochastic network optimization, energy harvesting and smart grid systems.
\end{IEEEbiography}

\begin{IEEEbiography}[{\includegraphics[width=1in,height=1.25in,clip,keepaspectratio]{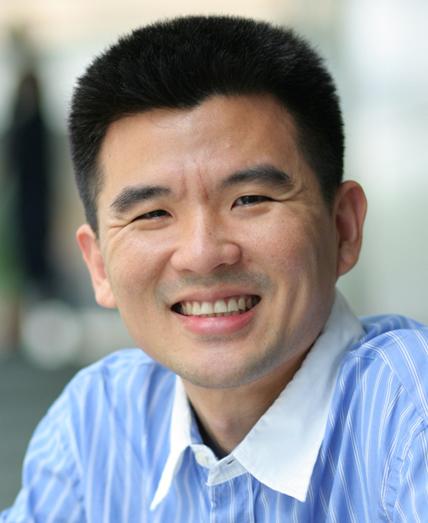}}]
{Tony Q.S. Quek}(S'98-M'08-SM'12) received the B.E.\ and M.E.\ degrees in Electrical and Electronics Engineering from Tokyo Institute of Technology, Tokyo, Japan, respectively. At Massachusetts Institute of Technology (MIT), Cambridge, MA, he earned the Ph.D.\ in Electrical Engineering and Computer Science. Currently, he is an Assistant Professor with the Information Systems Technology and Design Pillar at Singapore University of Technology and Design (SUTD). He is also a Scientist with the Institute for Infocomm Research. His main research interests are the application of mathematical, optimization, and statistical theories to communication, networking, signal processing, and resource allocation problems. Specific current research topics include cooperative networks, heterogeneous networks, green communications, smart grid, wireless security, compressed sensing, big data processing, and cognitive radio.

Dr.\ Quek has been actively involved in organizing and chairing sessions, and has served as a member of the Technical Program Committee as well as symposium chairs in a number of international conferences. He is serving as the TPC co-chair for IEEE ICCS in 2014, the Wireless Networks and Security Track for IEEE VTC Fall in 2014, the PHY \& Fundamentals Track for IEEE WCNC in 2015, and the Communication Theory Symposium for IEEE ICC in 2015. He is currently an Editor for the {\scshape IEEE Transactions on Communications}, the {\scshape IEEE Wireless Communications Letters}, and an Executive Editorial Committee Member for the {\scshape IEEE Transactions on Wireless Communications}. He was Guest Editor for the {\scshape IEEE Communications Magazine} (Special Issue on Heterogeneous and Small Cell Networks) in 2013 and the {\scshape IEEE Signal Processing Magazine} (Special Issue on Signal Processing for the 5G Revolution) in 2014.

Dr.\ Quek was honored with the 2008 Philip Yeo Prize for Outstanding Achievement in Research, the IEEE Globecom 2010 Best Paper Award, the 2011 JSPS Invited Fellow for Research in Japan, the CAS Fellowship for Young International Scientists in 2011, the 2012 IEEE William R. Bennett Prize, and the IEEE SPAWC 2013 Best Student Paper Award.
\end{IEEEbiography}

\begin{IEEEbiography}[{\includegraphics[width=1in,height=1.25in,clip,keepaspectratio]{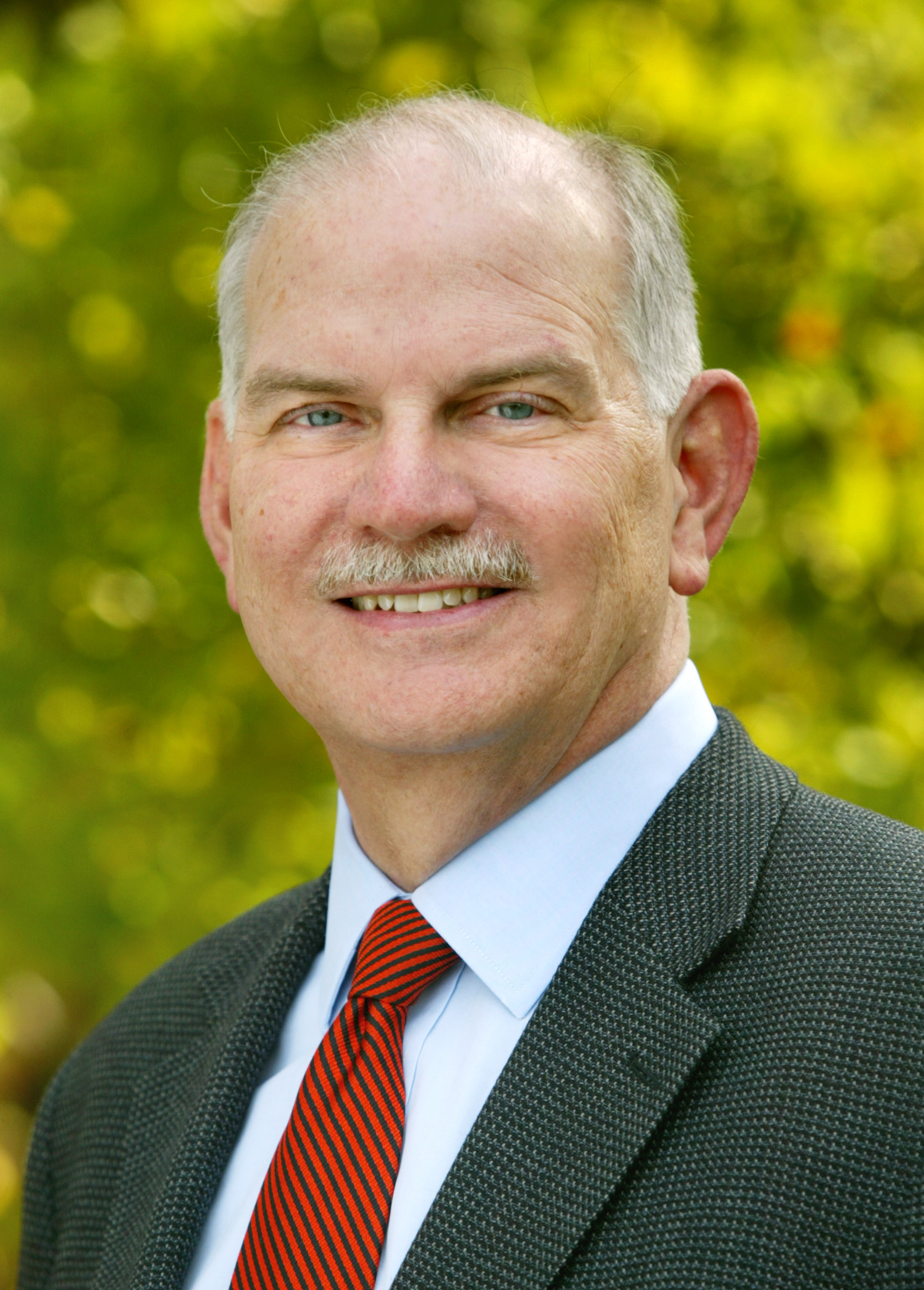}}]
H. Vincent Poor (S’72, M’77, SM’82, F’87) received the Ph.D. degree in EECS from Princeton University in 1977.  From 1977 until 1990, he was on the faculty of the University of Illinois at Urbana-Champaign. Since 1990 he has been on the faculty at Princeton, where he is the Michael Henry Strater University Professor of Electrical Engineering and Dean of the School of Engineering and Applied Science. Dr. Poor's research interests are in the areas of information theory, statistical signal processing and stochastic analysis, and their applications in wireless networks and related fields including social networks and smart grid. Among his publications in these areas are the recent books \emph{Principles of Cognitive Radio} (Cambridge University Press, 2013) and \emph{Mechanisms and Games for Dynamic Spectrum Allocation} (Cambridge University Press, 2014).

Dr. Poor is a member of the National Academy of Engineering, the National Academy of Sciences, and Academia Europaea, and is a fellow of the American Academy of Arts and Sciences, the Royal Academy of Engineering (U. K), and the Royal Society of Edinburgh. He received the Marconi and Armstrong Awards of the IEEE Communications Society in 2007 and 2009, respectively. Recent recognition of his work includes the 2010 IET Ambrose Fleming Medal for Achievement in Communications, the 2011 IEEE Eric E. Sumner Award, and honorary doctorates from Aalborg University, the Hong Kong University of Science and Technology, and the University of Edinburgh.  
\end{IEEEbiography}

\end{document}